\documentclass{ifacconf}
\usepackage{xcolor}
\usepackage{amsmath}
\usepackage{amsfonts}
\usepackage{amssymb}
\usepackage{tikz}
\usepackage{afterpage}
\usepackage{placeins}
\usepackage{bm}
\usepackage{paralist}
\usepackage{graphicx}
\usepackage{comment}
\usepackage{float}
\usepackage{algorithm}
\usepackage{algpseudocode}
\usepackage{natbib}
\algrenewcommand\textproc{}% change function titles to \textsc
\newtheorem{assumption}{\textbf{Assumption}}
\newtheorem{problem}{\textbf{Problem}}
\newtheorem{definition}{\textbf{Definition}}
\newtheorem{lemma}{\textbf{Lemma}}
\newtheorem{remark}{\textbf{Remark}}

\begin{document}
\begin{frontmatter}

\title{A Generalized Nash Equilibrium-Seeking Scheme for Trauma Resuscitation } 
% Title, preferably not more than 10 words.

%\thanks[footnoteinfo]{Sponsor and financial support acknowledgment
%goes here. Paper titles should be written in uppercase and lowercase
%letters, not all uppercase.}

\author[Ekpo]{Promise Ekpo} 
\author[Ekpo]{Angelique Taylor} 
\author[Molu]{Lekan Molu}

\address[Ekpo]{Cornell University, NY.  \{poe6, amt298\}@cornell.edu.}
\address[Molu]{Bala Cynwyd, PA. lekanmolu@scriptedonachip.com}

%%%%%%%%%%%%%%%%%%%%%%%%%%%%%%%%%%%%%%%%%%%%%%%%%%%%%%%%%%%%%%%%%%%%%%%%%%%%%%%%

\begin{abstract}               
% Angelique's updates
Trauma resuscitation is a clinical process for treating life-threatening physiological disorders in safety-critical environments, driven by the experience of healthcare workers (HCWs). 
Designing and optimizing quantifiable metrics that accurately capture HCW decisions may augment current resuscitation procedures with the potential to improve patient outcomes.
This motivates our socio-technical formulation of trauma resuscitation as a distributed generalized Nash equilibrium (GNE)-seeking game with coupled inequality constraints. This method is optimized over a time-varying communication graph. We introduce novel insights from clinical experience to model HCWs behavior.  This work facilitates the best possible resuscitation outcome given HCWs' workloads, schedules, competencies, and limited resources.
\end{abstract}


%%%%%%%%%%%%%%%%%%%%%%%%%%%%%%%%%%%%%%%%%%%%%%%%%%%%%%%%%%%%%%%%%%%%%%%%%%%%%%%%

\end{frontmatter}

%%%%%%%%%%%%%%%%%%%%%%%%%%%%%%%%%%%%%%%%%%%%%%%%%%%%%%%%%%%%%%%%%%%%%%%%%%%%%%%%

\definecolor{light-blue}{rgb}{0.30,0.35,1}
\definecolor{light-green}{rgb}{0.20,0.49,.85}
\definecolor{purple}{rgb}{0.70,0.69,.2}

\newcommand{\lb}[1]{\textcolor{light-blue}{#1}}
\newcommand{\bl}[1]{\textcolor{blue}{#1}}

\newcommand{\maybe}[1]{\textcolor{gray}{\textbf{MAYBE: }{#1}}}
\newcommand{\inspect}[1]{\textcolor{cyan}{\textbf{CHECK THIS: }{#1}}}
\newcommand{\more}[1]{\textcolor{red}{\textbf{MORE: }{#1}}}
\providecommand{\sectionautorefname}{Section}
%\renewcommand{\sectionautorefname}{Sec.}
%\renewcommand{\equationautorefname}{equation}
%\renewcommand{\subsectionautorefname}{$\S$}
%\renewcommand{\subsubsectionautorefname}{$\S$}
%\renewcommand{\chapterautorefname}{Chapter}
%\renewcommand{\definitionautorefname}{Def.}

% FYA
\newcommand{\cmt}[1]{{\footnotesize\textcolor{red}{#1}}}%{#2}
\newcommand{\todo}[1]{\textcolor{purple}{TO-DO: #1}}
\newcommand{\stopped}[1]{\color{red}STOPPED HERE #1\hrulefill}

%Text commands
\newcounter{mnote}
\newcommand{\marginote}[1]{\addtocounter{mnote}{1}\marginpar{\themnote. \scriptsize #1}}
\setcounter{mnote}{0}
\newcommand{\ie}{$i.e.$\ }
\newcommand{\eg}{e.g.\ }
\newcommand{\cf}{c.f.\ }
\newcommand{\yes}{\checkmark}
\newcommand{\no}{\ding{55}}

%Reference commands
\newcommand{\flabel}[1]{\label{fig:#1}}
\newcommand{\seclabel}[1]{\label{sec:#1}}
\newcommand{\tlabel}[1]{\label{tab:#1}}
\newcommand{\elabel}[1]{\label{eq:#1}}
\newcommand{\alabel}[1]{\label{alg:#1}}
\newcommand{\fref}[1]{\cref{fig:#1}}
\newcommand{\sref}[1]{\cref{sec:#1}}
\newcommand{\tref}[1]{\cref{tab:#1}}
\newcommand{\eref}[1]{\cref{eq:#1}}
\newcommand{\aref}[1]{\cref{alg:#1}}

\newcommand{\bull}[1]{$\bullet$ #1}
\newcommand{\argmax}{\text{argmax}}
\newcommand{\argmin}{\text{argmin}}
\newcommand{\mc}[1]{\mathcal{#1}}
\newcommand{\bb}[1]{\mathbb{#1}}

\def\tidx{t}
%\def\comment
%\def\value{V}
% from https://www.cs.jhu.edu/~jason/advice/write-the-paper-first.html
\newcommand{\Note}[1]{}
\renewcommand{\Note}[1]{\hl{[#1]}}  % comment out this definition to suppress all Notes
%\algnewcommand\algorithmicforeach{\textbf{for each}}
%\algdef{S}[FOR]{Foreach}[1]{\algorithmicforeach\ #1\ \algorithmicdo} %

%\newcolumntype{M}[1]{>{\centering\arraybackslash}m{#1}}
\def\coriolis{\textbf{\textit{C}}}
\def\massinertia{\textbf{\textit{M}}}
\def\torque{\bm{\tau}}
\def\frictionvec{\textbf{\textit{f}}}
\def\Smat{\textbf{\textit{S}}}
\def\sgn{\text{sgn}}
\def\Bmat{\textbf{\textit{B}}}
\def\wheelrad{\textbf{\textit{r}}}

\def\stateweight{\textbf{\textit{w}}_x}
\def\actionweight{\textbf{\textit{w}}_u}
\def\advactionweight{\textbf{\textit{w}}_v}

%Thesis defs
\def\kau{\mc{K}}
\def\particle{\textbf{x}}
\def\deformationgrad{\textbf{F}}
\def\refconf{\bm{\chi}_0}
\def\refconfbody{\mathscr{B}_0}
\def\conf{\bm{\chi}}
\def\currconf{\bm{\chi}}
\def\Eulerian{\mc{E}}
\def\cauchystress{\bm{\sigma}}
\def\stresscomp{\sigma}
\def\currconfbody{\mathscr{B}}
\def\strain{W}
\def\materialresponse{\textbf{G}}
\def\orthoggroup{{\textit{SO}}(3)}
\def\liegroup{{\mathbb{SE}}(3)}
\def\liealgebra{\mathfrak{se}(3)}
\def\identity{\textbf{I}}
\newcommand{\trace}[1]{\textbf{tr}(#1)}
\def\leftcauchy{\textbf{B}}
\def\rightcauchy{\textbf{C}}
\def\fiber{\textbf{dx}}

\def\dof{\text{DOF }}
\def\dofs{\text{DOFs }}
\def\reline{\mathbb{R}}
\def\curve{\deformationgrad}
\def\twist{{\xi}}
\def\contactforce{\tilde{F}}
\def\contactforcecomp{f}
\def\gaussianmap{\textbf{}n}
\def\contacttorquecomp{\tau}
\def\wrt{with respect to }
\def\curveparam{\position}
\def\basis{\bm{e}}
\def\pose{{g}}
\def\selmap{B}
\def\manipmap{{G}}
\def\jacob{\bm{J}}
\def\hinf{\mc{H}_\infty}
\def\htwo{\mc{H}_2}
\def\position{\textbf{r}}
\def\deformationgradcur{\textbf{H}}
\def\eulerianvel{\textbf{v}(\position, t)}
\def\headparam{\zeta}
\def\cspace{\mathcal{C}}

% mechanism defs
\def\wallthickness{1.5cm}
\def\sorodiam{5cm}
\def\sorodiamdim{5-6.25cm}

% inline macros
\newcommand{\putsoro}[2]{\includegraphics[width=.45\columnwidth,height=#2\columnwidth]{../../../PhDThesis/figures/#1}}
\newcommand{\sorowidth}{.35}
% \afterpage{
% %%%%%%%%%%%%%%%%%%%%%%%%%%%%%%%%%%%%%%%%%%%%%%%%%%%%%%%%%%%%%%%%%%%%%%%%%%%%%%%%
% \begin{figure}[t!]
% \centering
% \includegraphics[width=\columnwidth]{camreadyimages_v3/Fig2_Topology(2).png}
% \caption{Communication graph $\mc{G}(\mc{A})(k)$ at
% $k=199$. Nodes are $n=6$ healthcare workers; edges link pairs within
% radius $r=200$ ft. Player positions drift stochastically each iteration,
% yielding a time-varying neighbor structure consistent with Assumption~1.}
% \label{fig:topology}
% \end{figure}

% \begin{figure*}[!t]
% \centering
% \includegraphics[width=\textwidth]{camreadyimages2/Fig2_Topology.png}
% \caption{\textbf{Time-varying communication graph $\mc{G}(\mc{A})(t)$.}
% Snapshots at iterations $k = 0, 100, 199$ showing player positions and active edges within radius $r = 200$ ft.} 
% \label{fig:topology}
% \end{figure*}

\section{Introduction}
Trauma resuscitation (TR) is a collaborative, often conflicting, coordination effort among a team of healthcare workers (HCWs). It is often characterized by significant differential in 
\begin{inparaenum}[(i)]
	\item skillset;  
	\item mental and physical alacrity; and 
	\item intra-team communication skills. 
\end{inparaenum} 
These factors are shaped by HCWs' workplace experience and can degrade over time. The degradation  may spur burnout and fatigue --- undermining productivity; miscommunication  informs duplicated execution of tasks --- engendering redundancy. Across medical organizational hierarchies, through verbal and aural cues~\citep{sarcevic2012teamwork, ogundare_integrated_2025} HCWs can improve  decision quality if a systematic and quantifiable  algorithmic framework is integrated into workflows to track performance. Previous efforts have integrated autonomous systems to support team dynamics~\citep{TaylorHRI24, Tanjim_2025}, deliver supplies~\citep{e2022assistant}, or sanitize surgical rooms~\citep{sun2007port}.  

\begin{figure}[t!]
\centering
\includegraphics[width=\columnwidth]{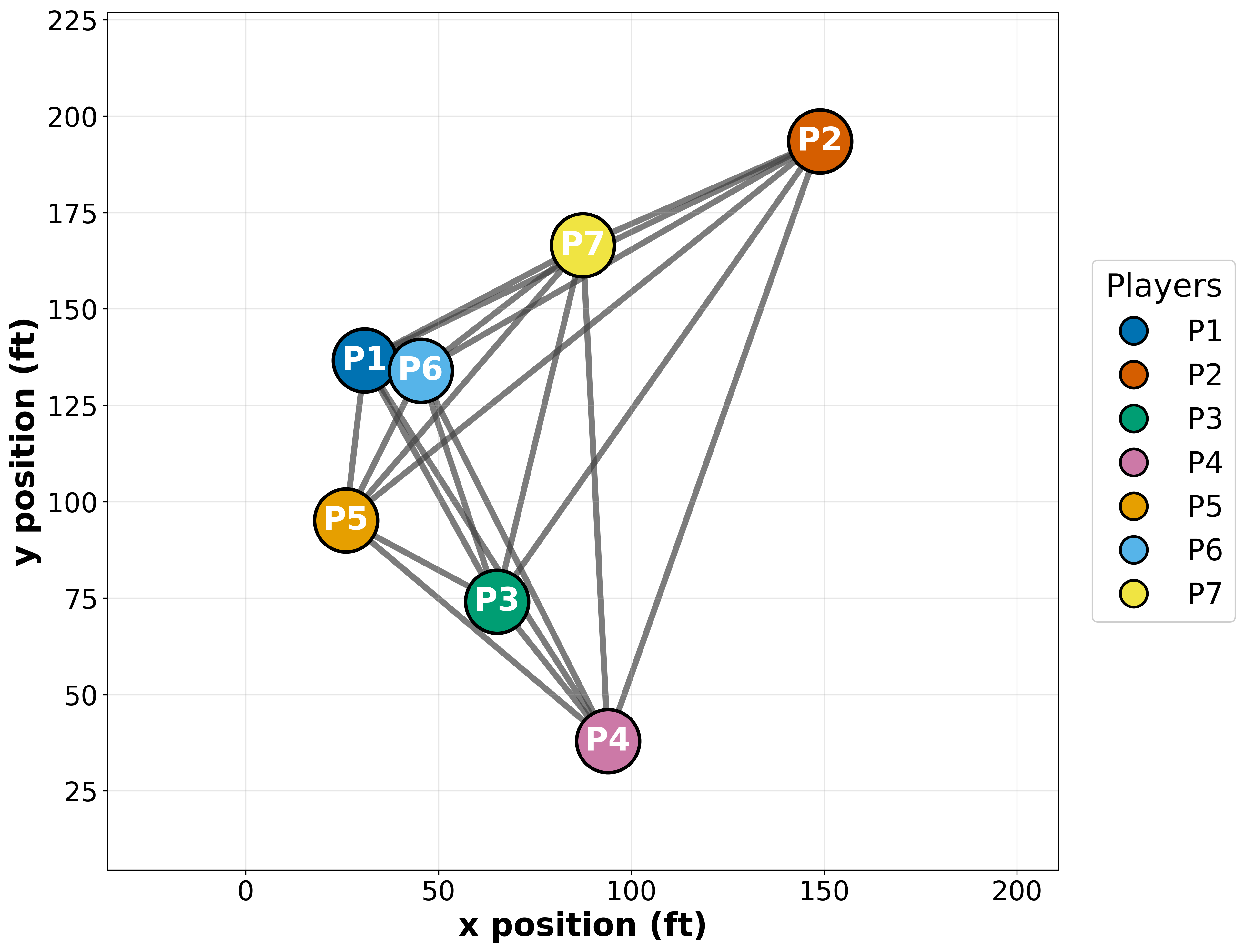}
\caption{Communication graph at
$k=0$. Nodes are $n=7$ healthcare workers; edges link pairs within
radius $r=200$ ft. Player positions yield a time-varying neighbor structure.}
\label{fig:topology}
\end{figure}
A well-executed TR procedure involves a  deliberative process where skill-informed task delegation by clinic leaders and responsibilities hand-off among peers are well-executed in spite of time constraints.
We set forth to design socio-technical process that
\begin{inparaenum}[(i)]
	\item supports clinical decision-making through automated task allocation among HCWs;
	\item maximizes HCW skill levels while counterbalancing other HCWs' skill levels; and
	\item enhances communication efficacy whilst reducing duplicate information exchange.
\end{inparaenum} 
This work studies quantifiable algorithms in high-stakes  socio-technical medical settings to enhance patient outcomes.  We take the view of a game-theoretic, collaborative TR setting, informed by our observations in a representative medical simulation training environment for HCWs at Weill Cornell Emergency Medicine in New York City\footnote{The BASE Camp event for interprofessional pediatric emergency medicine team training in the emergency care of critically ill and injured children was employed~\citep{WeillCornell}.}. 
We provide a mathematical characterization between observed skill levels and task models in the algorithmic framework we present.
% \subsection{Multiagent Healthcare Automation}
% \label{subsec:back::hcw_auto}

The role of shared mental models and decentralized communication workflows has been described in high-duress trauma resuscitation  contexts~\citep{sarcevic2012teamwork, Batra_Pioldi_Ekpo_Sayatqyzy_Maruur_Otieno_Ching_Taylor_2026}. Co-design studies have proposed cognitive aids and augmented-reality systems to improve information flow and situational awareness for distributed clinical teams~\citep{TaylorHRI24}. MATEC~\citep{cho_application_2025} and TraumaFlow~\citep{neumann_traumaflowdevelopment_2024} demonstrate that software agents can align team actions with clinical protocols  with centralized controllers. 
Distributed optimization in games over networks may be characterized with generalized Nash equilibrium (GNE)-seeking algorithms that better enhance  multi-agent socio-technical systems with coupling constraints~\citep{dave_social_2022}. 
% Operator-theoretic methods, including forward--backward splitting in monotone inclusions, provide scalable distributed algorithms with convergence guarantees~\citep{belgioioso_distributed_2021}.
% When agents interact over sparse networks and make decisions based on local observations, these tools yield consensus-seeking strategies with provable stability. 
Models incorporating agent dynamics, constraints, and communication topology align well with real-world healthcare teams. In these environments, agents exhibit  rationality under partial observability.

The rest of this paper is structured as follows: \S \ref{sec:background} describes the notations and games background. Our hypothesis is presented in \S \ref{sec:innovs} and it is evaluated on numerical test environments in \S \ref{sec:numerics}. We conclude the paper in \S \ref{sec:conclude}.

%%%%%%%%%%%%%%%%%%%%%%%%%%%%%%%%%%%%%%%%%%%%%%%%%%%%%%%%%%%%%%%%%%%%%%%%%%%%%%%%

\section{Background and Preliminaries}
\seclabel{background}

%Our work builds on recent advances in distributed GNE-seeking solution concepts, and applies embodied action prediction and generation in generative modeling to account for the complexities of fuzzy decision-making under confounding surrounding observations. 

%%%% Angelique commented out; add back is space permits
%This section describes healthcare automation in general multi-agent high-duress environments; which are then contextualized within distributed GNE (Sec. \ref{subsec:back::hcw_auto}). Trauma resuscitation as a distributed game is discussed in (Sec. \ref{subsec:back::trauma_resus}).  The mathematical notations used throughout the rest of the paper are presented (Sec. \ref{subsec:back::notations}) a a games taxonomy is presented in \ref{subsec:back::machinery}. Assumptions, definitions, lemmas and theorems used to build our results are given in (Sec. \ref{subsec:back:theoretical_machines}).

\subsection{General Notations} %Preliminaries}
\label{subsec:back::notations}

The set of real (non-negative) numbers is denoted, $\reline\, (\reline_+)$, the $m$-dimensional vector space is $\reline^m$, and the $n\times n$ dimensional real matrix is denoted $\reline^{n\times n}$. A time-varying variable, $x$ is denoted $x(t)$. For a function $f$ that depends on $x(t)$, we write $f(x; t)$. The absolute value of scalar $x$ is $|x|$. The Euclidean norm of the vector $x\in \reline^n$ is $\|x\|:=\sqrt{x^\top x}$. The $n$-dimensional vector of ones (zeroes) is  ${1}_n ({0}_n)$. For a differentiable function $J(x)$, the first-order derivative of $J(x)$ is denoted $\nabla_x J(x)$. For a set $\Omega$ and variable $x$, the Euclidean projection of $x$ onto $\Omega$ is $\mc{P}_\Omega(x) = \argmin_{x\in \Omega} \|x - x'\|_2$. The set $\mc{S} \subset \reline^m$ is a convex set if for $\alpha >0$ and for every  $x,y \in \mc{S}$, $\alpha x + (1-\alpha) y \in \mc{S}$. The cardinality of a $\mc{S}$ is denoted $\lfloor \mc{S} \rfloor$. Suppose that $\mc
{S}$ is a closed convex set, and $\mc{P}_{\mc{S}}(x)$ is a projection of $x$  onto the set $\mc{S}$, then there is a unique element $\mc{P}_{\mc{S}}(x) \in \mc{S}$ such that $\|x - \mc{P}_{\mc{S}}(x)\| = \inf_{y \in \mc{S}} \| x- y\|$. Similarly, the Euclidean projection of x onto the set  $\Omega$ is denoted $P_{\Omega}(x)\| = \argmin_{y \in \Omega} \| x- y\|$,

\subsection{Games Taxonomy}
\label{subsec:back::machinery}
Let a game be denoted as $\Gamma(\mc{V}, \Omega, J)$, where $\mc{V}$ are the players $\{1, \ldots, n\}$, each with action $x_i$ belonging in a constraint set $\Omega_i \in \reline^m$ and a local cost profile $J_i: \Omega_i \rightarrow \reline$ so that $J = \{J_1, \cdots, J_n\}$. The action profile of the game is $\Omega:=\Omega_1 \times \cdots \times \Omega_n$. The vector formed by all players' actions is  $x :=(x_i, x_{-i}) \triangleq \{x_i\}_{i=1}^n$, where the vector  formed by all the players' actions except those of player $i$ is $x_{-i} = \{x_{i'}\}_{i'=1, \, i'\neq i}^n \in \reline^{n_{-i}}$ where $n_{-i}:=n-n_i$.  

Define $K_i^p:=\{x_i \in \Omega_i \, |\, h_i(x_i)  \leq 0\}$ as the set of each individual player's constraint set, and $K^s:=\{x \in \Omega \,|\, g(x) \ge 0_m\}$ as the set of shared constraints, where $g(x) = c^\top x - d, \,\text{for } c=[c_1^\top, \cdots, c_n^\top]^\top, \ c_i \in \reline^m $ and $d \in \reline$.  It follows that player $i$'s action is constrained along two directions: constraints that depend on the action of other players $\ie {g(x)} \ge  0$, and individual constraints that depend on player $i$'s action \ie $x_i \in K_i$. Contrary to popular formulations in GNE-seeking literature, we have chosen $g(x)\ge 0$ since the skills, time indices, and other parameters we are optimizing for cannot be negative.
% We reserve this for future work.
%
%\begin{align}
%	K_{-i} = \prod_{i^\prime = 1, i^\prime \neq i}^{n} K_{i^\prime}.
%\end{align}

\subsection{Games Machinery}
\label{subsec:back:theoretical_machines}

Let us set a few definitions and preliminary results in motion. 
Denote the feasible set of player $i$' action by $K_i(x_{-i}):=\{x_i \ | \ (x_i, x_{-i}) \in K\}$, where $K := K^s \cap (K_1 \times \cdots K_n)$.   Suppose we define the generalized Nash equilibrium problem (GNEP) as $\Gamma(\mc{V}, \Omega, J, K)$, then the individual players in the game ``freeze" other players' actions $x_{-i}$ as exogenous variables, to solve
\begin{align}
	 \min_{x_i \in \Omega_i} &J_i(x_i, x_{-i}), \quad \forall\,\, i \in \mc{V} \nonumber\\
	\text{subject to } &x_i \in K_i(x_{-i}).
	\tag{GNEP}
	\label{eq:GNEP}
\end{align}

\begin{definition}[Neighbors of a Player]
	We define the neighbors $\mc{N}_i(t)$ of player $i$ at time $t$ as the set of all agents that lie within a predefined radius, $r_i$. $\mc{N}_i(t) = \{j \mid (j,i) \in \mc{E}\}$.%, of agent $i$ at time $t$. 
\end{definition}

\begin{definition}[Generalized Nash Equilibrium]
	In problem \eqref{eq:GNEP} , the action profile $x^\star:=(x_i^\star, x_{-i}^\star)$ is a GNE if $J_i(x_i^\star, x_{-i}^\star) \le J_i(x_i, x_{-i}^\star), \, \forall \, x_i \in \Omega_i, \, \forall i \in \mc{V}$.
	\label{prob:GNEP}
\end{definition} Assumptions are made for feasibility of the solution \eqref{eq:GNEP}.
\begin{assumption}[Interconnectivity]
	The game occurs over a time-varying \textit{connected}\footnote{That is, any two nodes in $\mc{V}$ are connected by a path.} and undirected\footnote{That is, $(i, j) \in \mc{E} \Leftrightarrow (j,i) \in \mc{E}$.} weighted  communication graph $\mc{G}(\mc{A})(t): \{\mc{V}, \mc{E}\}$ where $t>0$ with edges $\mc{E} \subset \mc{V} \times \mc{V}$ and  adjacency matrix $\mc{A}:=[a_{ij}] \in \reline^{n \times n}$. Player $i$ is connected to (can receive information from)  player $j$ if $(i,j) \in \mc{E}$ and vice versa. Let $a_{ij} > 0$ if $(i,j) \in \mc{E}$ and $a_{ij} = 0$ otherwise. For a $D = diag(d_{11}, \cdots, d_{nn}) \in \reline^{n\times n}$, where $d_{ii} = \sum_{j=1}^n a_{ij}$ for all $i \in \mc{V}$, we have $\mc{G}(\mc{A})$'s Laplacian as $\mc{L} := D- \mc{A} \in \reline^{n\times n}$
	\label{ass:connect}
\end{assumption}

\begin{assumption}[Cost Function]
	For $i \in \mc{V}$, $\Omega_i \in \reline^m$ is a nonempty, convex, and closed set; and the cost $J_i(x_i, x_{-i})$ is continuously differentiable on $\Omega$ and convex in $x_i$ for every fixed $x_{-i}$. Define $\nabla J(x) = [\nabla_{x_1} J_1^\top(x), \cdots, \nabla_{x_n} J_n^\top(x)]^\top$ as the \textit{single-value} game mapping. 
	\label{ass:costfun}
\end{assumption}

\begin{assumption}[Monotonicity and Game-mapping]
	The single-valued mapping $\mc{M}: \reline^{m} \rightarrow \reline^{m}$ is strongly monotone on the action constraint set $\Omega \in \reline^m$ if there exists a constant $\mu >0$ for any $x, x^\prime \in \Omega, \text{ where } x\neq x^\prime$ such that  $(\nabla J -\nabla J^\prime)^\top (x-x^\prime) \ge \mu \|x - x^\prime\|^2$ for any $\nabla J \in \mc{M}(x)$, $\nabla J^\prime \in \mc{M}(x^\prime)$.
	\label{ass:monotonic}
\end{assumption}

\begin{assumption}[Constraint Function]
	\label{ass:constraint_fun}
	The inequality constraint function $g(x) \ge 0$ is continuously differentiable and convex in $x$. Also, the feasible action set $K$ is nonempty, convex, and closed --- satisfying the Slater's condition \ie, $x^\star \in \text{int}(\Omega)$\footnote{The relative interior of the set $\Omega$ is denoted $\text{int}(\Omega)$.} such that $g(x^\star) \le 0_m$.
\end{assumption}

\begin{lemma}[Existence of a GNEP]
	Let the GNEP of Def. \ref{prob:GNEP} be given and suppose that (i) $\Omega_i \in \reline^m$ be a nonempty, convex, and compact such that for  $i \in \mc{V}$, $K_i(x_{-i})$ is nonempty, convex, closed, and $K_i$ is both upper and lower semi-continuous; and (ii) the local cost $J_i(x_i, x_{-i})$ is quasi-convex on $K_i(x_{-i})$ for every $i \in \mc{V}$. Then, a GNE exists.
	\label{lem:gnep}
\end{lemma}
\begin{pf}(Existence of GNEP)
	Follows directly from Thm.~4.1 in \citep{facchinei2010generalized}.
\end{pf}
\begin{remark}
	When assumptions \ref{ass:costfun},  \ref{ass:monotonic}, and \ref{ass:constraint_fun} are fulfilled, then the GNEP $\Gamma(\mc{V}, \Omega, J, K)$ satisfies the existence condition  $x^\star$ in Lemma \ref{lem:gnep}.
\end{remark} 
\begin{lemma}[Optimality of the GNE]
	Suppose that there exists a multiplier $\lambda_i^\star \in \reline^m$, $\forall \, i \in \mc{V}$. The optimality condition for each player  can be found via the KKT conditions 
	\begin{align}
		\nabla_{x_i} J_i(x_i^\star, x_{-i}^\star)& + \langle \lambda_i^\star, \nabla_x g(x^\star) \rangle = 0_n, \,\, \forall \,\, i \in \mc{V}, \lambda_i^\star \in \mc{V}, \nonumber \\
		0_m \le & \,\lambda_i^\star \perp  \,\, g(x_i^\star, x_{-i}^\star) \ge 0_m
		\label{eq:variational_ineq} 
	\end{align}
	under the assumption of continuous differentiability. This follows from Thm.~4.6 in~\cite{facchinei2010generalized}.
	\label{lem:vi}
\end{lemma}
As a system of equations, \eqref{eq:variational_ineq}  can be written as 
\begin{align}
{J}&({x^\star, \lambda^\star}) = 0, \nonumber \\
0 \le\,\, &{\lambda^\star} \perp {g}({x}^\star) \,\,\ge \,\,0,  
\label{eq:gnep-kkt-vi}
\end{align} 
\begin{align}
\text{ for } &	{g(x^\star)} =  [g^\top(x_1^\star), \cdots, g^\top(x_n^\star)]^\top,\,
	{\lambda} =  [\lambda_1^{\star\top}, \cdots, \lambda_n^{\star\top}]^\top, \nonumber \\
 \text{and }&	{J}({x^\star, \lambda^\star}) =  [\nabla_{x_1} J_1^\top(x_1^\star),  \cdots,  \nabla_{x_n} J_n^\top(x_n^\star)]^\top. \nonumber
\end{align}

With the constraint qualification, the $\bm{x}$ portion of \eqref{eq:gnep-kkt-vi} is a first-order necessary condition for the GNEP; and under appropriate convexity assumptions,  the $\bm{x}$ part solves the GNEP so that \eqref{eq:gnep-kkt-vi} is sufficient condition for the GNEP.

\begin{remark}
	Lemma \ref{lem:vi} is the so-called variational inequality problem $\text{VI}(X, \nabla J(x))$ whose solution set is a special class of GNEs with all Lagrange multipliers equal. That is, the solution $x^\star$ to the $VI(X, \nabla J(x))$ under the KKT conditions \eqref{eq:variational_ineq} satisfies the GNE $x^\star$ if $\lambda_1^\star = \cdots = \lambda_n^\star = \lambda^\star$.
\end{remark}

%%%%%%%%%%%%%%%%%%%%%%%%%%%%%%%%%%%%%%%%%%%%%%%%%%%%%%%%%%%%%%%%%%%%%%%%%%%%%%%%

\section{A GNEP Trauma Resuscitation Scheme}
\seclabel{innovs}
In this section, we present our formalism of clinical TR workflows into an games and optimization problem.

%---------------------------------------------------------------------
\subsection{Clinical model}

\begin{table}[tb!]
\caption{Representative HCW skill attributes (ALS protocol).}
\centering
\begin{tabular}{|p{2.4cm}|l|}
\hline
\textbf{Taxonomy} & \textbf{Meaning} \\
\hline
CPR & Cardiopulmonary resuscitation \\
\hline
SHOCK & Deliver defibrillation shock \\
\hline
RHYTHM\_CHECK & Evaluate rhythm every 2 min \\
\hline
AIRWAY\_MGMT & Airway interventions \\
\hline
EPI & Epinephrine administration \\
\hline
TEAM\_COMM & Team communication \\
\hline
\end{tabular}
\label{tab:als}
\end{table}

Each healthcare worker is a player $i \in \mc{V}$ with objective $J_i$ (\eg, start CPR), individual constraints $K_i^p$ (\eg, hemorrhage control), and shared constraints $K_i^s$ (\eg, breathing or pulse checks) in a game $\Gamma(\mc{V}, \Omega, \{J_i\}, K)$ where $K_i = K_i^s \cap K_i^p$. Table~\ref{tab:als} itemizes typical HCW skill attributes inspired by the ALS code card for adult cardiac arrest~\citep{als_redcross}. The inequality constraint $h_i(x_i) \leq 0$ represents each player's clinical role assignment.

The state $x_i$ of each player is
\begin{align}
	x_i(t) := [\,s_i,\ a_i(t),\, f_i(t),\, \upsilon_i(t)\,]^\top
\end{align}
with skill proficiency index $s_i$ and time-dependent decision variables: alacrity index $a_i(t)$, fairness index $f_i(t)$, and communication efficiency $\upsilon_i(t)$. Here $s_i$ is a fixed parameter characterizing player $i$'s competence, while $[\,a_i(t), f_i(t), \upsilon_i(t)\,]$ constitutes player $i$'s action and the state of the dynamical system that converges to the Nash equilibrium. The action profile lies in $\Omega = K_1 \times \cdots \times K_n \subset \reline^{3n}$.

\noindent \textbf{The skill proficiency index} $s_i \in [0,1]$ is a fixed scalar averaging player $i$'s competence over past assignments and training (nursing, clinical certifications, physician training), calibrated by the institution. Performance degradation due to fatigue is captured indirectly through the time-varying alacrity index $a_i(t)$.

%---------------------------------------------------------------------
\noindent \textbf{Players' transient step response} is described as a second-order system~\citep{NormanNise}
\begin{align}
	\sigma_i(t) &= 1 - \frac{1}{\sqrt{1-\zeta^2}} \exp{(-\zeta \omega_n t)} \cos( \omega_n   \nonumber \\
	& \qquad \qquad\qquad\qquad\qquad\qquad\sqrt{1-\zeta^2}\,t - \phi)
	\label{eq:transient}
\end{align}
where $\phi = \arctan(\zeta/(\sqrt{1-\zeta^2}))$ for damping ratio $\zeta$ and natural frequency $\omega_n$. We use an underdamped model ($\zeta < 1$) since human physiological responses to stress overshoot before settling~\citep{Rosenblum_Rab_Admon_2025}. We adopt $\omega_n^i := 2\pi b_i / 60$ where $b_i$ is player $i$'s resting heart rate (bpm), a proxy for the engagement transient time.\footnote{In practice, sensor noise in $b_i$ would be handled by standard filtering (\eg, multisensor fusion, low-pass smoothing).
% and the projection $P_{\Omega_i}$ in~\eqref{alg:dist-gne-seek} confines states to feasible sets regardless.
}

\noindent \textbf{Alacrity measure}: The alacrity index $a_i$ encodes a player's readiness to execute tasks within their skillset. We model $a_i$ via the rise time $t_r$ of \eqref{eq:transient}, which is the time for $\sigma(t)$ to go from 10\% to 90\% of its final value. Fixing $\zeta = k < 1$ and using $\omega_n^i t$ as the normalized time variable, $a_i(t)$ is found as
\begin{align}
	a_i(t) = \Delta \sigma(t) / \omega_n^i, \quad t \ge 0.
\end{align}

\noindent \textbf{Fairness in workload distribution} should be encouraged when players share resources \citep{ekpo_skill-aligned_2025,Ekpo_Agarwal_Grimm_Molu_Taylor_2025}. If one worker bears workload beyond their capacity, the uneven throughput becomes untenable for efficient resuscitation. We adopt the Jain fairness index~\citep{jain_quantitative_1998}, a convex, scale-invariant metric, bounded in $[1/n, 1]$:
\begin{align}
	f_i(w; t)= \frac{\left(\sum_{i=1}^{n} w(t)\right)^2}{n\sum_{i=1}^{n} w(t)^2}, \quad w(t), t \ge 0
	\label{eq:jfi}
\end{align}
which measures the equity of work $w$ allocated to player $i$. Each player maintains a local team-level estimate $\hat f_i(t)$ obtained by running dynamic averaging consensus on $\sum_j w_j$ and $\sum_j w_j^2$. The per-task weights $\delta_\tau \in [0,1]$ that define the workload $w_i(t) = \sum_{\tau \in \mathcal{T}_i(t)}\delta_\tau$ are specified in Section~\ref{sec:numerics}.

\noindent \textbf{Communication efficiency.} We define the communication efficacy as a function of
\begin{inparaenum}[(i)]
	\item settling time;
	\item engagement level; and
	\item response time.
\end{inparaenum}
The \textbf{settling time} $t_s^i$ is the time for $\sigma_i(t)$'s damped oscillations to reach and stay within $\pm 2\%$ of steady state, approximated as
\begin{align}
	t_s^i = - \ln (\epsilon \sqrt{1-\zeta^2}) /\zeta \omega_n
\end{align}
with $\epsilon = 0.02$. The \textbf{response time} is the difference between reaching steady state and being within 10\% of $\sigma_i(t)$'s final value upon a task handover. The timely response is $\bar t_i = k\Delta t_{ij}$ where $\Delta t_{ij}$ is the inter-call time between players $i$ and $j$ for all $j \in x_{-i}$. The \textbf{engagement level} for player $i$ is a running mean of visit counts from neighbors $\mc{N}_i$ up to discretized time $\mc{K}$:
\[
	e_i (t) = \frac{1}{T n_{-i}}\int_{t_0}^{T} \left(\sum_{k=1}^{\mc{K}} \nu_{ij}(t)\right) dt ,
	\quad  \forall \, j \in x_{-i}.
\]
A ``visit'' is a discrete communication event (verbal handoff, task delegation, closed-loop confirmation) recorded at each timestep $k$. The \textbf{communication efficiency index} is
\begin{align}
	\upsilon_i (t) = 	\alpha(t_s^i + \bar{t}_i) + (1-\alpha) e_i(t) , \,\, \alpha \in (0, 1).
\end{align}

\subsection{Local costs and dynamics}

\begin{problem}
	We want a distributed algorithm for a GNEP $\Gamma(\mc{V}, \Omega, \nabla J(x), K)$ for the socio-technical game with coupled constraints, $K$.
\end{problem}
The local cost for player $i$ is the quadratic objective between player $i$ and its neighbors $\mc{N}_i$:
\begin{align}
	&J_i(x_i, x_{-i}; t) := J_i(x_i; t) + \langle J_i(x_{-i}; t) \rangle_{r_i}
\end{align}
where
\begin{align}
	J_i(x_i; t) &= \|s_i\|^2 + \|a_i (t)\|^2 + \|f_i(w; t)\|^2 +\|\upsilon_i(t)\|^2  \nonumber
\end{align}
\text{and}
\begin{align}
	 \langle J_i(x_{-i}; t) \rangle_{r_i}  &= \dfrac{1}{1+n_i(t)}\left(J_i(x_{i}; t) + \sum_{j \in \mc{N}_i(t)}^{} J_i(x_{-i}; t)\right)
\end{align}
denotes the cost of neighboring players within radius $r_i>0$ of player $i$. The resuscitation coordination is modeled as a dynamic game whose actions evolve under the distributed dynamics in~\eqref{alg:dist-gne-seek}, converging to a variational generalized Nash equilibrium (v-GNE) $x^* = (x_1^*, \ldots, x_n^*)$:

\begin{subequations}\label{alg:dist-gne-seek}
\begin{align}
\dot{x}_i &= P_{\Omega_i}\!\left(x_i - \alpha\nabla_{x_i} J_i (x) + \tfrac{\alpha\gamma}{n}\, \nabla_{x_i} g_i(x_i)^\top \lambda_i\right) - x_i \\[-2pt]
\dot{\lambda}_i &= P_{\mathbb{R}_+^p}\!\left[\kappa \sum_{j \in \mc{N}_i} \mathrm{sgn}(\lambda_j - \lambda_i) - g_i(x_i)\right] \\[-2pt]
\dot{\zeta}_i &= \rho \sum_{j \in \mc{N}_i} \mathrm{sgn}(\eta_j - \eta_i), \quad
\dot{\eta}_i = \zeta_i + J_i(x_{-i})
\end{align}
\end{subequations}

where $\alpha > 0$ is the primal step size; $\gamma > 0$ scales the
constraint coupling, normalized by $n$; and $\kappa, \rho > 0$ are
the consensus gains for the dual and aggregation channels. The
Lagrange multiplier $\lambda_i \in \mathbb{R}_+^p$ carries one
component per inequality constraint in $g_i(x_i)$; $\zeta_i$ and
$\eta_i$ are auxiliary scalars implementing finite-time dynamic
averaging consensus on $J_i(x_{-i})$. Initial conditions are
$x_i(0) \in \Omega_i$, $\lambda_i(0) \in \mathbb{R}_+^p$, and
$\zeta_i(0) = 0$. Solutions are interpreted similarly to Filippov~\citep{cortes2008discontinuous}. Algorithm~\eqref{alg:dist-gne-seek}
extends~\cite{liang2017distributed} to inequality constraints
$g(x) \ge 0$.
The constants are chosen as $\kappa > (n-1) h_1$ and $\rho > \gamma (n-1) h_2$, where
\[
h_1 = \sup_{i\in \mc{V}} \left(\sup_{x_i\in \Omega_i} \|\nabla_{x_i}J_i(\cdot, x_{-i})\| \sup_{y, y^\prime \in \Omega} \|y - y^\prime \| \right),
\]
\[
h_2 = \sup_{i \in \mc{V}} \left(\sup_{x_i\in \Omega_i} \|{g_i(x_i)}\|\right),
\]
computed distributively as in~\citep{liang2017distributed}.
%%%%%%%%%%%%%%%%%%%%%%%%%%%%%%%%%%%%%%%%%%%%%%%%%%%%%%%%%%%%%%%%%%%%%%%%%%%%%%%%

\section{Numerical Results}
\seclabel{numerics}

\begin{figure*}[t!]
\centering
\includegraphics[width=\textwidth]{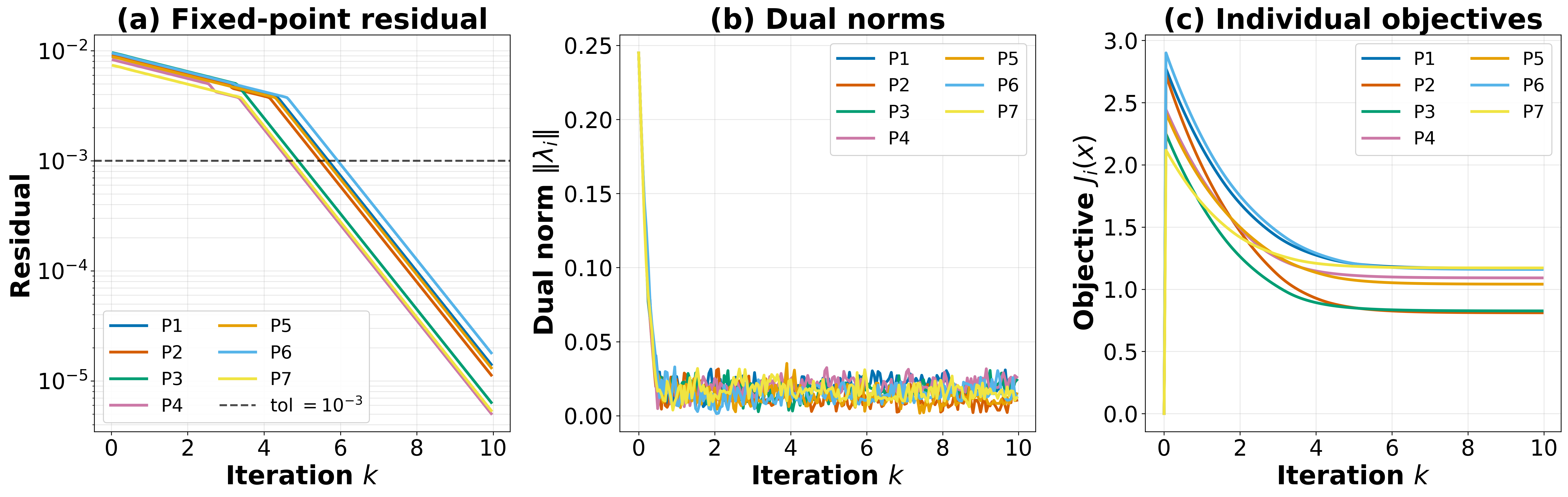}
\caption{Distributed GNEP simulation for 7 players.
(a) Fixed-point residual.
(b) Dual variable norms.
(c) Individual objectives.}
\label{fig:gnep_results}
\end{figure*}

\noindent\textbf{Setup.}
We validate the distributed GNEP algorithm on a trauma resuscitation
scenario with $7$ healthcare workers executing the Advanced Life Support
(ALS) cardiac arrest protocol~\citep{als_redcross}.
The synthetic dataset models a $20$-minute episode with clinically accurate
CPR cycles, rhythm checks, shock delivery, drug administration, airway
management, and Return of Spontaneous Circulation (ROSC) assessment.
Each timestamp assigns one atomic action per agent (\eg, \texttt{Start\_CPR},
\texttt{Deliver\_Shock}, \texttt{Administer\_Epinephrine}), following the
ideal ALS flowchart.

\noindent\textbf{Initial conditions and constraints.}
Workers are initialized with heterogeneous state components:
alacrity in $[0.5,1.0]$, fairness in $[0.8,0.95]$, communication efficiency
$\upsilon_i$ in $[0.9,1.0]$, and fixed skill proficiencies in $[0.6,0.9]$.
A spatial communication radius $r=200$~ft induces time-varying neighbor
sets $\mathcal{N}_i(t)$.
The shared constraint $g(x) \ge 0$ enforces minimum operational thresholds
$a_i, f_i, \upsilon_i \ge 0.2$, preventing excessive fatigue or coordination
degradation. All states are normalized to $[0,1]$.

\noindent\textbf{Task weighting.}
Each task $\tau$ is assigned a weight $\delta_\tau \in [0,1]$, with negative
penalties applied when a player's assignment falls below the
protocol-prescribed threshold: $-1.0$ for suboptimal role assignment, $-0.5$
for low alacrity, and $-0.5$ for low energy. The negative values discount
the workload contribution of substandard assignments so that the resulting
workload $w_i(t) = \sum_{\tau \in \mathcal{T}_i(t)} \delta_\tau$ enters the
Jain fairness index in~\eqref{eq:jfi} with the right sign.

\noindent\textbf{Convergence.}
Figure~\ref{fig:gnep_results} shows the simulation results.
Fixed-point residuals decay below the $10^{-3}$ tolerance within roughly
six iterations (panel a), confirming that all players reach a stable
joint decision under the time-varying graph. Dual norms remain bounded
throughout (panel b), indicating that the constraints remain
appropriately enforced without divergence. Individual objectives reach
heterogeneous steady values consistent with the KKT conditions of a
v-GNE (panel c); the spread across players reflects each player's
fixed skill $s_i$, since each one optimizes its own cost rather than
a shared team cost.

\noindent\textbf{Communication topology.}
Figure~\ref{fig:topology} shows the communication graph
at $k=0$, with seven healthcare workers within a $200$-ft
radius. Edges connect player pairs whose Euclidean distance falls
below the radius. Since player positions drift after each
iteration, the neighbor sets $\mathcal{N}_i(k)$ change over time, and
the graph in subsequent iterations differs from the initial snapshot. The dynamics in~\eqref{alg:dist-gne-seek} operate on
the graph $\mathcal{G}(\mathcal{A})(k)$ at each iteration.

% \noindent\textbf{Baseline.}
% The baseline is a centralized greedy task-assignment scheme that assigns
% each pending ALS atomic action to the player with the highest remaining
% $s_i \cdot a_i(t)$ score. It ignores coupled constraints.

\noindent\textbf{Scalability.}
Figure~\ref{fig:scalability} shows how the algorithm scales with team
size $n$. The number of iterations to reach the $10^{-3}$ convergence threshold
depends on how well-connected the communication graph is. As long as
the graph stays well-connected as $n$ grows, the iteration count grows
only mildly.
Empirically, this count grows mildly with $n$, from $\sim 124$ at $n=5$
to $\sim 147$ at $n=50$. Wall-clock time per iteration is dominated by
the neighbor sums and scales as $O(n + |\mathcal{E}|)$, well below the
$O(n^2)$ worst-case reference.

\begin{figure}[tbp]
\centering
\includegraphics[width=\columnwidth]{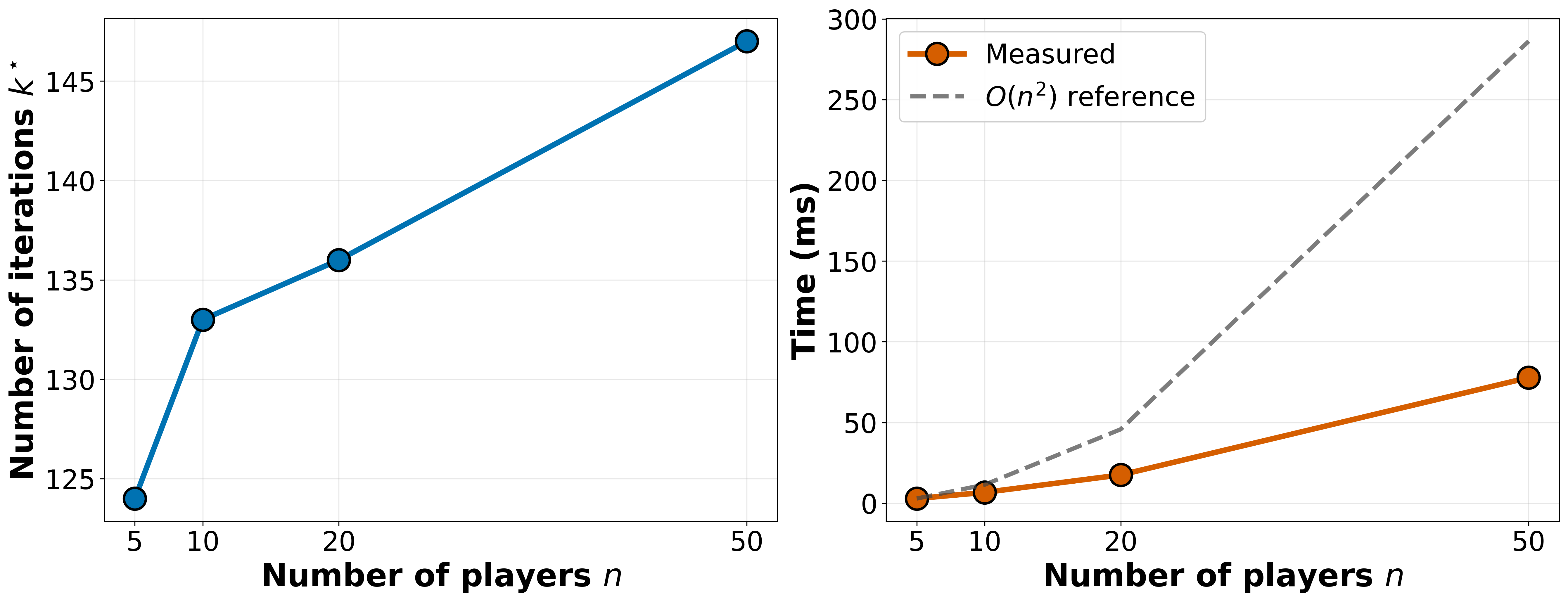}
\caption{Scalability of the distributed GNEP algorithm.
(a) Iteration count to reach the convergence threshold $10^{-3}$
versus $n$. (b) Wall-clock time versus $n$, with $O(n^2)$ worst case reference
(dashed).}
\label{fig:scalability}
\end{figure}

\noindent\textbf{Chattering analysis.}
The signum function in~\eqref{alg:dist-gne-seek} is discontinuous,
which can cause small oscillations in numerical simulation. To
mitigate these oscillations, we smooth the signum with a
``smooth-abs'' function
\begin{align}
	\mathrm{sgn}_\alpha(x) := \sqrt{x^2 + \alpha^2} - \alpha,
	\quad \alpha > 0,
	\label{eq:smooth-abs}
\end{align}
following~\citet{tassa2012synthesis}. The original discontinuous
dynamics are interpreted in the standard Filippov framework
of~\citet{cortes2008discontinuous}. Figure~\ref{fig:chattering}
compares the two. Smoothing preserves convergence to the v-GNE
and eliminates the oscillations.

\begin{figure}[tbp]
\centering
\includegraphics[width=\columnwidth]{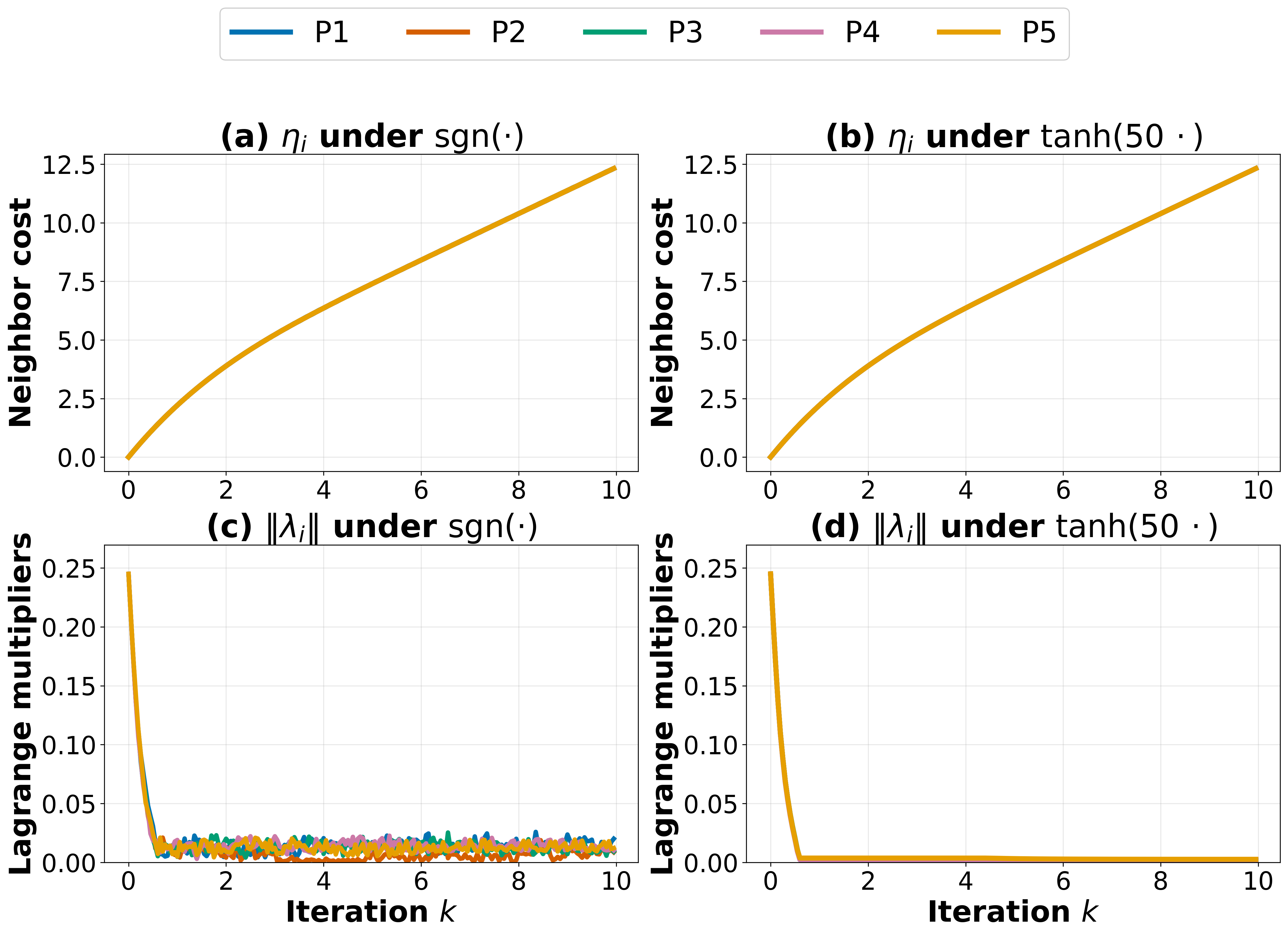}
\caption{Discontinuous vs.\ smoothed signum, $n=5$ players.
(a, b) Neighbor cost under $\mathrm{sgn}(\cdot)$ and
$\tanh(50\,\cdot)$. (c, d) Lagrange multipliers under each.}
\label{fig:chattering}
\end{figure}

%%%%%%%%%%%%%%%%%%%%%%%%%%%%%%%%%%%%%%%%%%%%%%%%%%%%%%%%%%%%%%%%%%%%%%%%%%%%%%%%
\FloatBarrier

\section{Conclusion}
\seclabel{conclude}
This paper formulated trauma resuscitation as a distributed generalized
Nash equilibrium problem in which healthcare worker skill, alacrity, fairness,
and communication efficiency enter as coupled decision variables. A
primal--dual seeking algorithm with inequality constraints was proposed
and analyzed. Simulations on a seven-player Advanced Life Support scenario
demonstrated convergence with bounded dual variables, and a sweep over up
to 50 players characterized the dependence of iteration count on team size.
This is a feasibility study on a synthetic clinical workflow that does not
capture protocol deviations or stress-induced behavior in the decision-making
process.
%%%%%%%%%%%%%%%%%%%%%%%%%%%%%%%%%%%%%%%%%%%%%%%%%%%%%%%%%%%%%%%%%%%%%%%%%%%%%%%%

%\begin{ack}
%Place acknowledgments here.
%\end{ack}

%\section*{DECLARATION OF GENERATIVE AI AND AI-ASSISTED TECHNOLOGIES IN THE WRITING PROCESS}
%During the preparation of this work the author(s) used [NAME TOOL / SERVICE] in order to [REASON]. After using this tool/service, the author(s) reviewed and edited the content as needed and take(s) full responsibility for the content of the publication.
\FloatBarrier
% \clearpage  
\bibliography{ifacconf} % bib file to 
%\appendix
%\section{A summary of Latin grammar}    % Each appendix must have a short title.
%\section{Some Latin vocabulary}              %
\end{document}